\documentclass[letterpaper, 12pt]{article}

\usepackage[left=2.6cm, right=2.6cm, top=2cm, bottom=1.8cm]{geometry}                                 
\usepackage{graphicx,color,array,subfigure,hyperref,url}
\usepackage{amsmath,amsthm,amssymb,amsfonts,booktabs,multirow,bm,balance}
\usepackage{caption}
\usepackage[normalem]{ulem}
\usepackage[noadjust]{cite}

\usepackage{algorithm,algcompatible}

\theoremstyle{plain}
\newtheorem{theorem}{Theorem}
\newtheorem{lemma}{Lemma}
\newtheorem{proposition}{Proposition}
\newtheorem{corollary}{Corollary}

\newtheorem*{claim*}{Claim}

\newtheorem*{example*}{Example}

\theoremstyle{definition}

\theoremstyle{remark}

\newtheorem*{remark*}{Remark}

\newcommand{\R}{\mathbb{R}}

\newcommand{\bbP}{\mathbb{P}}

\usepackage{accents}

\makeatletter
\newsavebox\myboxA
\newsavebox\myboxB
\newlength\mylenA
\newcommand*\xoverline[2][0.75]{%
    \sbox{\myboxA}{$\m@th#2$}%
    \setbox\myboxB\null
    \ht\myboxB=\ht\myboxA%
    \dp\myboxB=\dp\myboxA%
    \wd\myboxB=#1\wd\myboxA
    \sbox\myboxB{$\m@th\overline{\copy\myboxB}$}
    \setlength\mylenA{\the\wd\myboxA}
    \addtolength\mylenA{-\the\wd\myboxB}%
    \ifdim\wd\myboxB<\wd\myboxA%
      \rlap{\hskip 0.5\mylenA\usebox\myboxB}{\usebox\myboxA}%
    \else
        \hskip -0.5\mylenA\rlap{\usebox\myboxA}{\hskip 0.5\mylenA\usebox\myboxB}%
    \fi}
\makeatother

\usepackage{booktabs}
\newcommand{\ra}[1]{\renewcommand{\arraystretch}{#1}}

\begin{document}
\title{Near-Optimal Rapid MPC using Neural Networks: \\ A Primal-Dual Policy Learning Framework}
\author{Xiaojing Zhang$^{\dagger,*}$, Monimoy Bujarbaruah$^{\dagger,*}$ and Francesco Borrelli\thanks{$\dagger$ authors contrinuted equally to this work. The authors are with the Model Predictive Control Laboratory, Department of Mechanical Engineering, University of California at Berkeley, CA 94720, USA (email: xgeorge.zhang@gmail.com, monimoyb@berkeley.edu, fborrelli@berkeley.edu). }
}

\maketitle

\begin{abstract}
In this paper, we propose a novel framework for approximating the explicit MPC policy for linear parameter-varying systems using supervised learning. Our learning scheme guarantees feasibility and near-optimality of the approximated MPC policy with high probability.
Furthermore, in contrast to most existing approaches that only learn the MPC policy, 
we also learn the ``dual policy", which enables us to keep a check on the approximated MPC's optimality online during the control process.
If the check deems the control input from the approximated MPC policy \emph{safe and near-optimal}, then it is applied to the plant, otherwise a backup controller is invoked, thus filtering out (severely) suboptimal control inputs.
The backup controller is only invoked with a bounded (low) probability, where the exact probability level can be chosen by the user.
Since our framework does not require solving any optimization problem during the control process, 
it enables the deployment of MPC on resource-constrained systems. Specifically, we illustrate the utility of the proposed framework on a vehicle dynamics control problem. Compared to online optimization methods, we demonstrate a speedup of up to 62x on a desktop computer and 10x on an automotive-grade electronic control unit, while maintaining a high control performance.
\end{abstract}


\section{Introduction}
Model Predictive Control (MPC) is an advanced control strategy that is able to optimize a system's behavior while respecting system constraints. Originating in process control, MPC has found application in fields such as building climate control \cite{Oldewurtel:energyBuildings:2012, MaMatuskoBorrelli2014,zhang:schildbach:sturzenegger:morari:13}, quadcopter control \cite{BouffardAswani2012, Mueller2013_QuadcMPC, DAndrea_2016ICRA}, self-driving vehicles \cite{GrauGaoLinHedrickBorrelli2013, rosolia2017autonomousrace, bujarbaruahtorque,  ZhangParking_CDC2018} and robotics in general \cite{PoignetGautier2000NMPCManipulator, Wieber2006LMPCWalking,  LevineAbbeel2016learning, 2018arXiv180604335B}. However, implementing MPC on fast dynamical systems with limited computation capacity remains generally challenging, since MPC requires the solution of an optimization problem at each sampling step. This is especially true for mass-produced systems, such as drones and automotive vehicles where low-powered and cheap processors are used.

Over the past decade, significant research effort has been devoted to enabling MPC to systems with limited computation power by developing numerically efficient solvers that exploit the structure of the MPC optimization problem \cite{RaoWrightRawlings1998_IPM_MPC, richter2009real,  MattingleyBoydCVXGen2012, FerreauQPOases2014, osqp}. Another approach to reduce computation load of MPC is to precompute the optimal control law offline, store the control law into the system, and evaluate it during run-time. This approach, known as {Explicit MPC}, is generally well-understood for linear time-invariant system where the optimal control law has been shown to be piecewise affine over polyhedral regions \cite{BorrelliBemporadMorari_book}. The main drawbacks of Explicit MPC is that synthesis of the optimal control law can be computationally demanding even for medium-sized problems, and storing and evaluating the look up tables can be prohibitive for embedded platforms \cite{KvasnicaFikarTAC2012}. To address this issue, {suboptimal} explicit MPC methods have been developed that are defined over fewer polyhedral regions and hence can be evaluated more efficiently, at the cost of reduced performance \cite{JohansenGrancharovaTAC2003, multiresSeanSummers, JonesMorariPolytopicApproximationTAC2010, KvasnicaFikarClippingTAC2012, KvasnicaAutomatica2013}. 

An alternative way of approximating explicit MPC control laws is by means of {function approximation}, where the control law is encoded, for example, through a Deep Neural Network (DNN) \cite{ParisiniNNAutomatica1995, ChenMorariACC2018, chen2019large}, or represented as a linear combination of basis functions \cite{DomahidiLearning2011, ultraFastBemporad}. In general, the goal is to find a function out of a given function class that best explains some given training data, using tools such as supervised learning \cite{ParisiniNNAutomatica1995, DomahidiLearning2011} or reinforcement learning \cite{ChenMorariACC2018}. The main advantage of such \emph{policy learning} methods is two-fold, namely $(i)$ approximators such as DNNs are believed to be able to encode complex functions with relatively few parameters \cite{ChenMorariACC2018, LuciaEfficientRepresentationArXiv2018, HertneckAllgowerLCSS2018}, and $(ii)$ while the training process can be computationally demanding, evaluating the control laws can be often carried out very efficiently, since it only requires evaluation of (simple) functions.
The main challenge with designing control policies with function approximators is that it is often not possible to verify feasibility and optimality of the encoded control laws, without introducing significant conservatism during the learning phase. 
This conservatism can be partially alleviated by resorting to methods that use statistical verification to ensure safety and performance. These methods, however, only provide probabilistic guarantees \cite{ParisiniNNAutomatica1995, HertneckAllgowerLCSS2018}. 
To the best of the authors' knowledge, such methods lack the ability to verify safety and performance deterministically and in real-time. {In other words, given a current state, these methods cannot certify safety and performance of the control input to be applied.} Nevertheless, in practice it is critical to verify a control input before it is applied to a system to ensure safety, performance and stability. 

To address this issue, this paper proposes a novel policy learning scheme, called \emph{primal-dual policy learning}, for approximating the explicit MPC control law for linear \emph{parameter-varying} systems, where we learn a \emph{primal policy} and a \emph{dual policy}. 
The primal policy encodes the approximated explicit MPC law, while the dual policy encodes an optimality certificate, which we use in real-time during the control process to estimate the performance of the control input given by the primal policy.
Specifically, our contributions can be summarized as follows:
\begin{itemize}
    \item We propose a novel supervised learning framework where we formulate our primal-dual policy learning problem as randomized optimization problems, whose constraints are defined by the training data set. Using tools from statistical learning theory and scenario optimization, we show that if the sample size of the training data is chosen appropriately, then the learned primal and dual policies satisfy the constraints and are near-optimal with high probability.
    Explicit sample sizes are provided when the policies are approximated through Deep Neural Networks or through the linear combination of basis functions.

    \item Using arguments from strong duality of convex optimization, we show that the dual policy can be used to estimate the (sub)optimality of the approximated MPC law that is encoded in the primal policy. Since evaluating the dual policy is computationally cheap, it  provides an efficient way of estimating the suboptimality of the approximated MPC law in real-time during the control process. 
    To the best of our knowledge, our framework is the first of its kind that not only provides probabilistic feasibility and optimality guarantees of the approximated explicit MPC policy, but is able to efficiently detect suboptimal control inputs, thus extending the works of \cite{ ParisiniNNAutomatica1995, Lucia_NMPC2018_DNNrobNMPC, HertneckAllgowerLCSS2018}.

    \item 
    Our methodology can be applied to linear parameter-varying systems and, more generally, MPC problems that require the solution of a convex optimization problem, hence extending the  existing literature on policy learning. 
    Furthermore, we show that the proposed framework is applicable to generic policy learning and function approximation schemes, including (deep) neural networks, linear combinations basis functions and kernel regression. 
    
    \item  We demonstrate the efficacy of the proposed primal-dual policy learning framework on an integrated chassis control problem in vehicle dynamics. In particular, we demonstrate computation speedups of up to {62x on a laptop computer, and up to 10x} on an automotive Electronic Control Unit (ECU), enabling the implementation of MPC on such mass-produced embedded systems.
\end{itemize}

The outline of the paper is as follows: Section~\ref{sec:problemDescription} describes the problem setup under consideration. Section~\ref{sec:background} reviews basic concepts of supervised learning, duality theory and randomized optimization. The proposed primal-dual policy scheme is presented and discussed in Section~\ref{sec:primalDualPolicyLearning}. We illustrate our method on an integrated chassis control example in Section~\ref{sec:simResults}. Finally, concluding remarks are presented in Section~\ref{sec:conclusion}.


\section{Problem Description}\label{sec:problemDescription}
\subsection{Dynamics, Constraints, and Control Objective}
We consider linear parameter-varying (LPV) systems of the form
\begin{equation}\label{eq:paramVaryingSystem}
    x_{t+1} = A(q_t)x_t + B(q_t)u_t,
\end{equation}
where $x_t\in\R^{n_x}$ is the state at time $t$, $u_t\in\R^{n_u}$ is the input at time $t$, and $A(q_t)$ and $B(q_t)$ are known matrices of appropriate dimensions, that depend on a time-varying parameter $q_t$. Throughout this paper, we assume that the parameter $q_t$ is known or can be predicted perfectly, for all time steps $t$. The system is subject to linear input and state constraints\footnote{Our framework can be extended to handle generic convex constraints; linearity is assumed here for simplicity.} of the form
\begin{equation}\label{eq:constraints}
    \mathbb U := \{u: H_u u_t \leq h_u\}, \qquad \mathbb X := \{x: H_x x_t \leq h_x\},
\end{equation}
for given $H_u$, $h_u$, $H_x$ and $h_x$.
The control objective is to minimize, over a finite horizon $T$, a quadratic cost of the form
\begin{equation}\label{eq:objFunction}
    x_T^\top Q_T x_T + \sum_{t=0}^{T-1} x_t^\top Q_s x_t + u_t^\top R_s u_t,
\end{equation}
where the matrices $Q_s$, $Q_T$ are assumed to be positive semi-definite, and $R_s$ is assumed positive definite.

\subsection{Model Predictive Control}
At each time step $t$, Model Predictive Control (MPC) measures the state $x_t$, obtains the parameters $[q_{0|t},\ldots,q_{T-1|t}]$, and solves the following finite-horizon optimal control problem
\begin{equation} \label{eq:MPCproblemOrig}
    \begin{array}{llll}
        \displaystyle \min_{U_t, X_t} & x_{T|t}^\top Q_T x_{T|t} + \displaystyle \sum_{k=0}^{T-1} x_{k|t}^\top Q_s x_{k|t} + u_{k|t}^\top R_s u_{k|t}  \\
        \ \ \mathrm{s.t.} & x_{k+1|t} = A(q_{k|t}) x_{k|t} + B(q_{k|t}) u_{k|t}, \\                    & (x_{k|t}, u_{k|t} ) \in \mathbb X \times \mathbb U,\  x_{T|t} \in \mathbb X_f, \\
                    & x_{0|t} = x_t,\ \ k=0,\ldots,T-1,
    \end{array}
\end{equation}
where $x_{k|t}$ is the state at time $t+k$ obtained by applying the predicted inputs $u_{0|t},\ldots,u_{k-1|t}$ to system~\eqref{eq:paramVaryingSystem}. Furthermore, $U_t := [u_{0|t} ,\ldots, u_{T-1|t}]$ and $X_t := [x_{0|t} ,\ldots, x_{T|t}]$ are the collection of all predicted inputs and states, respectively. The set $\mathbb X_f \subset\R^{n_x}$, which we assume is a compact polytope, is a so-called terminal set, and ensures recursive feasibility of the MPC controller, see e.g.\ \cite{BorrelliBemporadMorari_book} for details. If $U^*_t$ is a minimizer of \eqref{eq:MPCproblemOrig}, then MPC applies the first input $u_t = u^*_{0|t}$. This process is repeated at the next time step, resulting in a receding horizon control scheme.

By eliminating the states $X_t$ from \eqref{eq:MPCproblemOrig}, we can express \eqref{eq:MPCproblemOrig} compactly as
\begin{equation}\label{eq:MPC_primal}
    \begin{array}{rlll}
        J^*(P_t) := \displaystyle \min_{U} & \frac{1}{2} U^\top Q(P_t) U + c(P_t)^\top U  \\
        \ \ \mathrm{s.t.} & H(P_t) U \leq h(P_t),
    \end{array}
\end{equation}
where $P_t := [x_t,q_{0|t},\ldots,q_{T-1|t}]$ is the collection of all parameters $\{q_{k|t}\}_k$ and the initial state $x_t$, and $Q(P_t)$, $c(P_t)$, $H(P_t)$, $h(P_t)$ are appropriately defined matrices, see e.g.,~\cite{Goulart2006} for their construction. 
Problem~\eqref{eq:MPC_primal} is a \emph{multi-parametric quadratic program}, whose optimizer $U^*(\cdot)$ and optimal value $J^*(\cdot)$ depend on $P_t$ \cite{BorrelliBemporadMorari_book}. To streamline the upcoming presentation, we assume in this paper that the parameters $P_t$ take values in a known compact set $\mathcal{P}$, i.e., $P_t\in\mathcal{P}$ for all $t$, and that \eqref{eq:MPC_primal} is feasible and finite for all $P_t \in \mathcal{P}$.

\subsection{Overview of Proposed Approach}
Solving the optimization problem~\eqref{eq:MPC_primal} at each sampling time can be computationally challenging for fast systems that run on resource-constrained platforms, even if modern numerical algorithms are used. In these cases, it is often computationally more efficient to precompute and store the optimal control law $U^*(\cdot)$ as a function of $P\in\mathcal{P}$, and then evaluate it in real-time when $P_t$ is known (``explicit MPC") \cite{BorrelliBemporadMorari_book}. Unfortunately, computing the explicit control law $U^*(\cdot)$ to problem~\eqref{eq:MPC_primal} is challenging due to the dependence of the system matrices $A(\cdot)$ and $B(\cdot)$ on the parameters $q_{k}$\footnote{If the system matrices are constant and time-invariant, then $U^*(\cdot)$ is known to be piecewise affine over polyhedral \cite{BemporadMorari1999, BorrelliBemporadMorari_book} and can be computed with existing algorithms \cite{MPT3}.}. 
To address this issue, we propose the use of function approximation to find a policy $\tilde U_\theta(\cdot) \approx U^*(\cdot)$ that approximates the optimal input. 
The approximated policy $\tilde U_\theta(\cdot)$ is computed offline (``learning"), and online, the control input is obtained by evaluating the approximated policy for the specific parameters $P_t$. 
The main challenge in using these tools is to ensure that the approximated policy returns a safe (i.e., constraints are satisfied) and high-performance (i.e., near-optimal) control input. 

In this paper, we propose a novel framework for learning the MPC policy that verifies feasibility and near-optimality in two stages: Offline, during the training phase, we propose the use of a supervised learning scheme, formulated as a randomized optimization problem, for learning the approximated MPC policy $\tilde U_\theta(\cdot) \approx U^*(\cdot)$. We show that, if the number of training data is chosen appropriately, then the approximated control law satisfies the system constraints and is near-optimal with high probability. Online, during the deployment phase, we propose a novel run-time \emph{deterministic} verification scheme, which, for a given parameter $P_t$ at time $t$, verifies safety and estimates the performance of $\tilde U_\theta(P_t)$. This run-time verification scheme is carried out during the control process using an auxiliary \emph{dual policy}.


\section{Technical Background}\label{sec:background}
In this section we briefly review the concept of function approximation (Section~\ref{sec:supLearning}), duality theory (Section~\ref{sec:dualitySection}), and randomized optimization (Section~\ref{sec:randOpt}), which serve as building blocks of our primal-dual policy learning scheme.

\subsection{Supervised Learning}\label{sec:supLearning}
The goal in function approximation is to approximate a function $f\in\mathcal{F}$, defined on some function space $\mathcal{F}$, by another function $\tilde f\in\tilde{\mathcal{F}}\subset \mathcal F$ such that $\|f - \tilde f\|_\mathcal{F}$ is ``small", where $\|\cdot\|_\mathcal{F}$ is some norm. Since this minimization problem is performed over the infinite dimensional space of functions $\tilde{\mathcal F}$, it is generally intractable. A common approach is to restrict oneself to function spaces $\tilde{\mathcal F}=\tilde{\mathcal F}_\theta$ that are characterized by a finite number of parameters $\theta$, and to approximate the norm $\|\cdot\|_{\mathcal F}$ by the empirical error. This is typically achieved by drawing $M$ samples $\{z^{(i)},f(z^{(i)})\}_{i=1}^M$ (``training data"), such that the finite-dimensional problem of \emph{learning $f(\cdot)$} is given by
\begin{equation}\label{eq:Function_Approx}
    \theta^* := \displaystyle \arg\min_\theta\ \ \sum_{i=1}^M \mathcal{L}\left( f(z^{(i)}) - \tilde{f}_\theta(z^{(i)}) \right), 
\end{equation}
where $\theta^*$ denotes the optimal parameter and $\mathcal{L}(\cdot)$ a loss function, such as the euclidean norm. The choice of the loss function and the function space $\tilde{\mathcal{F}}_\theta$ is often problem-dependent. In the following, we describe two popular supervised learning algorithms, and we refer the interested reader to \cite{Smola2004,Genton:2002} for details.

\subsubsection*{Deep Neural Networks}
A deep neural network (DNN) of $L$ layers is a parametric function approximator of the form 
\begin{equation*}
    \tilde f_\theta(z) = \left( \lambda_L \circ \sigma \circ \lambda_{L-1} \circ \ldots \circ \sigma \circ \lambda_1 \right) (z),
\end{equation*}
where $\lambda_i(z) := W_i z + b_i$ are affine functions with parameters $(W_i,b_i)\in\R^{w_i\times w_{i-1}}\times\R^{w_i}$, $\theta := \{(W_i,b_i)\}_{i=1}^L$ is the collection of all parameters, and $w_i$ is the width of the $i$th layer, and is a design parameter. The activation function $\sigma(\cdot)$ is also a design parameter, and popular ones are Rectified Linear Units (ReLU) $\sigma(z) = \max\{0,z\}$, sigmoids $\sigma(z)=1/(1+e^{-z})$, or hypertangent $\sigma(z) = \tanh\{z\}$. Neural networks are popular functions approximators because of their universal function approximation property, i.e., under (mild) technical assumptions, neural networks can represent any continuous function \cite{cybenko1989approximation, lu2017expressive}. 

\subsubsection*{Weighted Sum of Basis Functions}
Another way of obtaining a finite-dimensional parametrization of the function space $\tilde{\mathcal F}_\theta$ is to restrict $\tilde f_\theta(\cdot)$ to be the linear combination of $L$ (potentially nonlinear) basis functions $\{\kappa_i(\cdot)\}_{i=1}^L$, such that $\tilde f_\theta(z) = \sum_{i=1}^L \theta_i \kappa_i(z)$. This approach can be interpreted as a two-layer neural network with $\sigma(\cdot) = [\kappa_1(\cdot),\ldots,\kappa_L(\cdot)]$ and $\lambda_1(\cdot)=\textnormal{id}(\cdot)$ is the identity map. Popular basis functions include (Chebyshev) polynomials bases, Fourier bases and  kernel functions.

\subsection{Duality Theory}\label{sec:dualitySection}
Duality is used in optimization to certify optimality of a given (candidate) solution. Specifically, to every convex optimization problem $p^* := \min_x\{f(x): h(x)\leq 0\}$, we can associate its \emph{dual problem} $d^* := \max_{\lambda}\{g(\lambda): \lambda\geq0\}$, where $g(\lambda) := \min_x\{f(x) - \lambda^\top h(x)\}$. Under appropriate technical assumptions\footnote{These include feasibility, finite optimal value, and constraint qualifications, see \cite{boyd2004convex} for details.}, it can be shown that $p^* = d^*$ (``strong duality"). In this paper, we will make use of the \emph{weak duality} property that, for every primal feasible point $h(x)\leq0$ and every dual feasible point $\lambda\geq0$, it holds
\begin{equation}\label{eq:weakDuality}
    g(\lambda) \leq f(x).
\end{equation}
Relationship \eqref{eq:weakDuality} can be used to bound the suboptimality of a (feasible) candidate solution $\bar x$, since $f(\bar x) - p^* \leq f(\bar x) - g(\lambda)$, for any $\lambda\geq0$. In Section~\ref{sec:primalDualPolicyLearning} of this paper, we will make use of this relationship to estimate the suboptimality of our approximated MPC policy.

\subsubsection*{Dual of \eqref{eq:MPC_primal}}
It is well-known that the dual of a convex quadratic optimization problem is again a convex quadratic optimization problem \cite{boyd2004convex}. Specifically, the dual of \eqref{eq:MPC_primal} takes the form
\begin{equation}\label{eq:MPC_dual}
    \begin{array}{rlll}
        D^*(P_t) := \displaystyle \max_{\lambda_t} & \frac{1}{2} \lambda_t^\top \tilde Q(P_t) \lambda_t + \tilde c(P_t)^\top \lambda_t + \tilde g(P_t)  \\
        \mathrm{s.t.}\ & \lambda_t \geq 0,
    \end{array}
\end{equation}
where $\tilde Q(\cdot)=\tilde Q(\cdot)^\top \preceq 0$, $\tilde c(\cdot)$ and $\tilde{g}(\cdot)$ depend on $P_t$. Similar to \eqref{eq:MPC_primal}, the optimizer of \eqref{eq:MPC_dual} depends on the parameters $P_t$, i.e, $\lambda_t^* = \lambda_t^*(P_t)$. Furthermore, since \eqref{eq:MPC_primal} is convex, it follows from strong duality that $J^*(P_t) = D^*(P_t)$, for all $P_t \in \mathcal{P}$.

\subsection{Randomized Optimization}\label{sec:randOpt}
Given a random variable $\delta\in\Delta$ on a probability space $(\Omega,\mathcal F, \mathbb P)$, consider a generic chance-constrained optimization problem
\begin{equation}\label{eq:StochOpt}
    \begin{array}{rlll}
        \displaystyle \min_{z \in \mathcal Z \subset \R^n} & J(z)  \\
        \ \ \mathrm{s.t.} \quad & \mathbb P[g(z,\delta) \leq 0]\geq1-\epsilon,
    \end{array}
\end{equation}
where $\epsilon\in(0,1)$ is the violation probability, $J\colon\R^n\to\R$ is the cost function, and $g:\R^n \times \Delta \to \R$ is the constraint function that depends on the uncertainty $\delta$. Problem~\eqref{eq:StochOpt} is in general computationally intractable to solve, and hence must be approximated in all but the simplest cases \cite{nemirovski:shapiro:06, Nemirovski2006}. 
One way of approximating \eqref{eq:StochOpt} is by means of randomization, where the chance constraint is replaced with a finite number of sampled constraints, i.e.,
\begin{equation}\label{eq:SampledOpt}
    \begin{array}{rlll}
        \displaystyle \min_{z \in \mathcal Z \subset \R^n} & J(z)  \\
        \ \ \mathrm{s.t.} \quad & g(z,\delta^{(i)}) \leq 0, \quad i=1,\ldots,N,
    \end{array}
\end{equation}
where $N>0$ is the so-called \emph{sample size}, and $\{ \delta^{(1)},...,\delta^{(N)} \}$ are i.i.d.~realizations of $\delta$. 
The fundamental question in randomized algorithms is the appropriate sample size $N$ such that the solution of the sampled program satisfies the chance constraints in \eqref{eq:StochOpt} with high probability. Generally speaking, two distinct paradigms exist for determining the sample size: \emph{scenario optimization} for convex sampled problems \cite{calafiore:campi:06, campi:garatti:08, calafiore:10, calafiore:fagiano:12, schildbach:fagiano:morari:13, ZhangAutomatica2015, grammatico:zhang:margellos:goulart:lygeros:16, Nasir_Care_Weyer_TCST2018}, and \emph{statistical learning learning theory} for non-convex sampled problems with finite {VC-dimension} \cite{anthony:biggs, vidyasagar:97, defarias:vanroy:04, erdogan:iyengar:06, tempo:calafiore:dabbene}. In the following subsections, we briefly review both concepts since they will be useful later on when determining the  training set size:

\subsubsection{Scenario Optimization}\label{ssec:scen_opt}
The theory of scenario optimization deals with problems for which \eqref{eq:SampledOpt} is convex. Specifically, under mild technical assumption\footnote{These typically include convexity and  compactness of the set $\mathcal Z$; convexity of $g(\cdot,\delta)$ and almost-sure feasibility of \eqref{eq:SampledOpt} for any sample size $N$ \cite{campi:garatti:08, calafiore:10}},
it was shown in \cite{campi:garatti:08}[Theorem~1] that the probability (also called ``confidence") of the optimizer $z^*$ of the sampled problem \eqref{eq:SampledOpt} not satisfying the chance constraint $\bbP[g(z^*,\delta)\leq0]\geq1-\epsilon$ is upper bounded by the beta distribution function $B_{n,\epsilon}(N) := \sum_{i=0}^{n-1} {N \choose i}\epsilon^i(1-\epsilon)^{N-i}$, where $n$ is the dimension of the decision variable, i.e., $\bbP^N\left\{ \bbP[g(z^*,\delta)>0] > \epsilon   \right\} \leq \sum_{i=0}^{n-1} {N \choose i}\epsilon^i(1-\epsilon)^{N-i}$.
Given a confidence level $\beta\in(0,1)$,  if the sample size satisfies $N \geq \frac{2}{\epsilon}(n-1 + \log\frac{1}{\beta})$, then $\bbP^N\left\{ \bbP[g(z^*,\delta)>0] > \epsilon   \right\} \leq \beta$ \cite{calafiore:10, alamo:tempo:luque:ramirez:Automatica15}. This formula shows that the confidence $\beta$ enters the sample size logarithmically, and can hence be chosen very small ($10^{-5}\sim 10^{-8}$) in practice.
We close by remarking that the linear dependency of $N$ on $n$ can be improved by taking into account additional problem structure, see \cite{schildbach:fagiano:morari:13, ZhangAutomatica2015, Zhang_PhD2016} for details.

\subsubsection{Statistical Learning Theory}\label{ssec:sst_bounds}
The sample size results in scenario optimization only apply to convex sampled problems. For non-convex sampled problems, one generally has to resort  to the theory of \emph{statistical learning} \cite{vapnik, vidyasagar:01, tempo:calafiore:dabbene},
where the sample sizes depend on a quantity called \emph{Vapnik-Chervonenkis} (VC) dimension $\xi\in\mathbb{N}$, rather than the number of decision variables. 
Roughly speaking, the VC-dimension captures the richness of the family of functions $\{\delta\mapsto g(x,\delta): z\in\mathcal{Z}\}$, see \cite{tempo:calafiore:dabbene}[Def.~10.2] for details. 
Assuming that $\xi$ is finite and given a desired confidence level $0<\beta \ll 1$, it has been shown in \cite{anthony:biggs}[Theorem~8.4.1] that, 
if $N \geq \frac{4}{\epsilon} \left( \xi \ln\frac{12}{\epsilon} + \log\frac{2}{\beta}    \right)$, then any feasible solution of the sampled problem \eqref{eq:SampledOpt} satisfies the chance constraint with confidence at least $1-\beta$.
Notice that, for a fixed confidence $\beta$, the sample size in statistical learning theory scales as $\mathcal{O}(\frac{\xi}{\epsilon}\log\frac{1}{\epsilon})$, compared to $\mathcal{O}(\frac{n}{\epsilon})$ for scenario programs.

\section{Primal-Dual Policy Learning}\label{sec:primalDualPolicyLearning}
In this section, we present our \emph{primal-dual policy learning} framework, where we learn two policies, the MPC policy $\tilde{U}_{\theta_p}(\cdot) \approx U^*(\cdot)$ (``primal policy"), and an auxiliary policy $\tilde\lambda_{\theta_d}(\cdot) \approx \lambda^*(\cdot)$ (``dual policy"). We show that $(i)$ probabilistic feasibility and optimality guarantees can be obtained offline by appropriately choosing the size of the training set, and $(ii)$ that the dual policy can be used online to efficiently estimate the suboptimality of the approximated MPC policy $\tilde{U}_{\theta_p}(\cdot)$.

\subsection{Primal and Dual Learning Problems}
We use supervised learning to approximate the primal policy $U^*(\cdot)$, which is the (parametric) minimizer of \eqref{eq:MPC_primal}, and dual policy $\lambda^*(\cdot)$, which is the (parametric) minimizer of \eqref{eq:MPC_dual}. To this end, we generate $N_p$ primal training data $\{P^{(i)}  , J^*(P^{(i)})  \}_{i=1}^{N_p}$, and $N_d$ dual training data $\{P^{(i)}, D^*(P^{(i)})\}_{i=1}^{N_d}$, where $P^{(i)}\in\mathcal{P}$ are samples randomly extracted according to a user-chosen distribution $\mathbb P$, and $J^*(P^{(i)})$ and $D^*(P^{(i)})$ are the optimal values obtained by solving \eqref{eq:MPC_primal} and \eqref{eq:MPC_dual}, respectively\footnote{Due to strong duality, $J^*(P) = D^*(P)$ for any parameter $P\in\mathcal{P}$}. The choice of the distribution $\mathbb P$, according to which the samples are extracted, depends on the task\footnote{For example, one could bias the distribution around a nominal operating point. If no such operating point is known, then a uniform distribution over the parameter space $\mathcal{P}$ is a natural choice.}.
To simplify forthcoming discussion, we define
\begin{subequations}
\begin{align}
     p(P;U) &:= \textstyle \frac{1}{2}U^\top Q(P)U + c(P)^\top U \label{eq:primalObj}\\
     d(P;\lambda) &:= \frac{1}{2} \lambda^\top \tilde Q(P)\lambda + \tilde{c}(P)^\top\lambda + \tilde g(P) \label{eq:dualObj},
\end{align}
\end{subequations}
where \eqref{eq:primalObj} is the objective function of the primal problem \eqref{eq:MPC_primal}, and \eqref{eq:dualObj} is the objective function of the dual problem \eqref{eq:MPC_dual}.
The \emph{primal learning problem} is now given by
\begin{subequations}
\begin{equation}\label{eq:primalLearning}
    \begin{array}{clll}
        \hspace{0.5cm} \displaystyle\min_{\theta_p, t_p} & t_p   \\
       \hspace{0.5cm} \ \ \text{s.t.} & 
       p(P^{(i)}; \tilde U_{\theta_p}(P^{(i)})) - J^*(P^{(i)}) \leq t_p , \\
       & H(P^{(i)})\tilde U_{\theta_p}(P^{(i)}) \leq h(P^{(i)}), \\     
       & \textnormal{for all } i = 1, \ldots, N_p,
    \end{array}
\end{equation}
and the \emph{dual learning problem} is given by
\begin{equation}\label{eq:dualLearning}
    \begin{array}{clll}
        \hspace{0.5cm} \displaystyle\min_{\theta_d, t_d} & t_d   \\
         \hspace{0.5cm}\ \ \text{s.t.} & 
        J^*(P^{(i)}) - d(P^{(i)}; \tilde\lambda_{\theta_d}(P^{(i)})) \leq  t_d, \\
        & \tilde \lambda_{\theta_d}(P^{(i)}) \geq 0, \\
        & \textnormal{for all } i = 1, \ldots, N_d,
    \end{array}
\end{equation}
\end{subequations}
where $\tilde U_{\theta_p}(\cdot)$ and $\tilde\lambda_{\theta_d}(\cdot)$ are user-defined functions with parameters $\theta_p$ and $\theta_d$, respectively (see Section~\ref{sec:aPrioriGuarantees} below).
Intuitively, \eqref{eq:primalLearning} and \eqref{eq:dualLearning} find the primal and dual policy that, with respect to the training data $\{P^{(1)},\ldots,P^{(N_{p/d})}\}$, minimize the worst-case  suboptimality $t_p$ and $t_d$, while ensuring feasibility of the policies. The above learning problems differ from most supervised learning schemes that learn $\tilde U_{\theta_p}(\cdot)$ by simply minimizing the average error $\min_{\theta_p} \sum_{i=1}^{N_p} \|  \tilde U_{\theta_p}(P^{(i)}) - U^*(P^{(i)}) \|^2$, and not enforcing feasibility and minimizing sub-optimality of the approximated policy.

Let $(\theta_p^*,t_p^*)$ and $(\theta_d^*,t_d^*)$ be the minimizers of \eqref{eq:primalLearning} and \eqref{eq:dualLearning}, respectively. In the remainder of this paper, we will refer to $\tilde U_{\theta_p^*}(\cdot)$ and  $\tilde \lambda_{\theta_d^*}(\cdot)$ as the \emph{approximated primal policy} and \emph{approximated dual policy}, respectively, and $t_p^*$ and $t_d^*$ as the primal and dual suboptimality estimates, respectively. 
Since the optimizers $(\theta_p^*,t_p^*)$ and $(\theta_d^*,t_d^*)$ depend on the realization of the training data, the approximated policies $\tilde U_{\theta_p^*}(\cdot)$ and  $\tilde \lambda_{\theta_d^*}(\cdot)$ are, strictly speaking,  random variables. 
In the following section, we provide  feasibility and optimality estimates of the approximated policies $\tilde U_{\theta_p^*}(\cdot)$ and $\tilde \lambda_{\theta_d^*}(\cdot)$.

\subsection{Probabilistic Safety and Performance Guarantees}\label{sec:aPrioriGuarantees}
In this section we show that by choosing the sample sizes appropriately in \eqref{eq:primalLearning}--\eqref{eq:dualLearning}, it is possible to provide \emph{probabilistic} feasibility and optimality guarantees of the approximated policies. Specifically, we will provide sample sizes for policies that are represented as $(i)$ a weighted sum of basis functions (Section~\ref{sec:SampleSizeBasisFunctions}), and $(ii)$ a (deep) neural network with rectified linear units (Section~\ref{sec:SampleSizesDNN}).

\subsubsection{Weighted Sum of Basis Functions}\label{sec:SampleSizeBasisFunctions}
We first consider the case when $\tilde U_{\theta_p^*}(\cdot)$ and $\tilde \lambda_{\theta_p^*}(\cdot)$ are parametrized as 
\begin{subequations}
\begin{align}
     \textstyle \tilde U_{\theta_p}(P) = \sum_{i=1}^{L_p} \theta_{p,i} \kappa_{p,i}(P) \label{eq:primalPolicyParam_linear}\\
     \textstyle \tilde \lambda_{\theta_d}(P) = \sum_{i=1}^{L_d} \theta_{d,i} \kappa_{d,i}(P) \label{eq:dualPolicyParam_linear},
\end{align}
\end{subequations}
where $\{\kappa_{p,i}(\cdot)\}_{i=1}^{L_p}$ are basis functions for the primal policy and $\{\kappa_{d,i}(\cdot)\}_{i=1}^{L_d}$ are the basis functions for the dual policy. The following proposition characterized sample sizes $N_p$ and $N_d$ such that the approximated policies are feasible and at most $t_{p/d}^*$ suboptimal with high probability.
\begin{proposition}\label{prop:scenarioOptSampleSize}
    Let $\epsilon_p,\epsilon_d\in(0,1)$ be admissible primal and dual violation probabilities, let  $0<\beta_p,\beta_d\ll1$ be desired primal and dual confidence levels, and assume that the primal and dual policies are parametrized as in \eqref{eq:primalPolicyParam_linear} and \eqref{eq:dualPolicyParam_linear}, respectively. If the primal and dual sample sizes satisfy $N_{p} \geq \frac{2}{\epsilon_p}(L_p + \log\frac{1}{\beta_p})$ and  $N_{d} \geq \frac{2}{\epsilon_d}(L_d + \log\frac{1}{\beta_d})$, then 
    $$\mathbb{P}\big[ H(P) \tilde U_{\theta_p^*}(P)\leq h(P), p(P; \tilde U_{\theta_p^*}(P)) \leq J^*(P) + t_p^* \big] \geq 1-\epsilon_p$$
    holds with confidence at least $1-\beta_p$, and  $$\mathbb{P}\big[ \tilde \lambda_{\theta_d^*}(P)\geq0,  d(P; \tilde \lambda_{\theta_d^*}(P)) \geq J^*(P)-t_d^* \big] \geq 1-\epsilon_d$$ holds with with confidence at least $1-\beta_d$, where $(\theta_p^*,t_p^*)$ and $(\theta_d^*, t_d^*)$ are the optimal solutions of \eqref{eq:primalLearning} and \eqref{eq:dualLearning}, respectively.
\end{proposition}
\begin{proof}
    The number of free variables in parametrization \eqref{eq:primalPolicyParam_linear} and \eqref{eq:dualPolicyParam_linear} is $L_p$ and $L_d$ respectively. Since $t_p\in\R$ and $t_d \in \R$, the total number of decision variables in \eqref{eq:primalLearning} and \eqref{eq:dualLearning} becomes $L_p+1$ and $L_d+1$ respectively. Using this in the bounds provided in Section~\ref{ssec:scen_opt} completes the proof.
\end{proof}
Intuitively, Proposition~\ref{prop:scenarioOptSampleSize} ensures that the primal (dual) policy is feasible and at most $t^*_p$ ($t^*_d$) sub-optimal  with probability at least $1-\epsilon_p$ ($1-\epsilon_d$). A forteriori,  Proposition~\ref{prop:scenarioOptSampleSize} estimates the performance (measured in terms of objective value) and safety (measured in terms of constraint satisfaction) of the approximated controller $\tilde{U}_{\theta_p}(\cdot)$.

\subsubsection{Deep Neural Networks}\label{sec:SampleSizesDNN}
Let us now consider the case when $\tilde U_{\theta_p^*}(\cdot)$ and $\tilde \lambda_{\theta_p^*}(\cdot)$ are parametrized through a (deep) neural network
\begin{subequations}
\begin{align}
     \textstyle \tilde U_{\theta_p}(P) =   \left( \lambda_{p,L_p} \circ \sigma \circ \lambda_{p,L_p-1} \circ \ldots \circ \sigma \circ \lambda_{p,1} \right) (P),
 \label{eq:primalPolicyParam_DNN}\\
     \textstyle \tilde \lambda_{\theta_d}(P) = \left( \lambda_{d,L_d} \circ \sigma \circ \lambda_{d,L_d-1} \circ \ldots \circ \sigma \circ \lambda_{d,1} \right) (P), \label{eq:dualPolicyParam_DNN}
\end{align}
\end{subequations}
where $\{\lambda_{p,i}(\cdot), \lambda_{d,i}(\cdot)\}_i$ are affine functions, $L_{p/d}$ are the depths, see also Section~\ref{sec:supLearning}. In this paper, we assume the use of Rectified-Linear Units (ReLU) \cite{Goodfellow-et-al-2016} activation functions $\sigma(\cdot)$\footnote{Our results can be extended to other activation functions as well, provided an estimate on the VC-dimension is known.}. 
We now have the following auxiliary lemma:
\begin{lemma}[\cite{bartlett2019nearly}]\label{lem:VCbound}
Assume that a neural network with ReLU activation units consists of $L$ layers, and that each layer has $n_\textnormal{act}^i$ activation units, and let $W$ be the total number of parameters. Then, the VC-dimension of such a neural network is upper bounded by
\begin{align*}
    \xi \leq\ & \bar\xi(L,W) \\[-1em]
    & := L + L W \log_2(4e \sum_{i=1}^{L } i n_{\textnormal{act}}^i \log_2(\sum_{i=1}^{L} 2e i n_{\textnormal{act}}^i)).
\end{align*}

\end{lemma}
We point the interested reader to \cite[Theorem~7]{bartlett2019nearly} for a more refined bound.

\begin{proposition}\label{prop:DNNSampleSize}
    Let $\epsilon_p\in(0,1)$ and $\epsilon_d\in(0,1)$ be admissible primal (dual) violation probabilities, let  $0<\beta_p\ll1$ and $0<\beta_d\ll1$ be desired primal and dual confidence levels. Assume that the primal and dual policies are parametrized as in \eqref{eq:primalPolicyParam_DNN} and \eqref{eq:dualPolicyParam_DNN}, respectively, with ReLU activation functions, and let 
    $W_{p/d}$ denote the total number of parameters $\theta_{p/d}$ in the primal/dual neural network.
    If the primal and dual sample sizes are chosen such that $N_p \geq \frac{4}{\epsilon_p} \left( \bar{\xi}(L_p,W_p) \ln\frac{12}{\epsilon_p} + \log\frac{2}{\beta_p}    \right)$ and  $N_d \geq \frac{4}{\epsilon_d} \left( \bar{\xi}(L_d,W_d) \ln\frac{12}{\epsilon_d} + \log\frac{2}{\beta_d}    \right)$, where $\hat{\xi}(\cdot,\cdot)$ is defined as in Lemma~\ref{lem:VCbound}, then
    $$\mathbb{P}\big[ H(P) \tilde U_{\theta_p^*}(P)\leq h(P), p(P; \tilde U_{\theta_p^*}(P)) \leq J^*(P) + t_p^* \big] \geq 1-\epsilon_p$$
    is satisfied with confidence at least $1-\beta_p$, and  $$\mathbb{P}\big[ \tilde \lambda_{\theta_d^*}(P)\geq0,  d(P; \tilde \lambda_{\theta_d^*}(P)) \geq J^*(P)-t_d^* \big] \geq 1-\epsilon_d$$ holds with with confidence at least $1-\beta_d$. 
\end{proposition}
\begin{proof}
The result is a direct consequence of Lemma~\ref{lem:VCbound} and Section~\ref{ssec:sst_bounds}.
\end{proof}
Similar to Proposition~\ref{prop:scenarioOptSampleSize}, Proposition~\ref{prop:DNNSampleSize} allows us to bound the probability of the approximated policies being infeasible and suboptimal. 

\subsection{Deterministic Feasibility and Optimality Certification}\label{sec:aPosterioriGuarantees}
One way of using the result from the previous section is to apply the approximated controller $\tilde{U}_{\theta_p}(\cdot)$ to  system \eqref{eq:paramVaryingSystem}, and, by virtue of Proposition~\ref{prop:scenarioOptSampleSize}, the control inputs will be feasible and at most $t^*_p$ suboptimal with probability at least $1-\epsilon_p$, i.e., for ``most of the time". 
The challenge of this approach, however, is that the feasibility and optimality statements are \emph{only probabilistic}; more precisely, for a given parameter $P_t\in\mathcal{P}$ at time $t$, Proposition~\ref{prop:scenarioOptSampleSize} is \emph{not} able to certify feasibility and optimality of $\tilde U_{\theta_p^*}(P_t)$.
This, however, is problematic in practice, where it is often desirable to attest feasibility and performance of the control input $\tilde U_{\theta_p^*}(P_t)$ before it is applied\footnote{A simple way of always obtaining feasibility is to project $\tilde U_{\theta_p^*}(\bar P)$ into the feasible set at each time step.}.

To address this issue, we next propose a ``deterministic" certification scheme for estimating optimality of $\tilde U_{\theta_p^*}(P_t)$. We begin with the following observation:
\begin{proposition}\label{prop:PDcheck}
    Given a parameter $P$, assume that $\tilde U_{\theta_p^*}(P)$ and $\tilde\lambda_{\theta_d^*}(P)$ are primal and dual feasible, respectively, i.e., $H(P)\tilde U_{\theta_p^*}(P) \leq h(P)$ and $\tilde\lambda_{\theta_d^*}(P)\geq0$. Then, the $\tilde U_{\theta_p^*}(P)$ is no more than
    \begin{equation}\label{eq:PDcheck}
            p(P;\tilde U_{\theta_p^*}(P)) - d(P;\tilde\lambda_{\theta_d^*}(P)),
    \end{equation}
    suboptimal, i.e., $p(P; \tilde U_{\theta_p^*}(P)) - J^*(P) \leq   p(P;\tilde U_{\theta_p^*}(x)) - d(P;\tilde\lambda_{\theta_d^*}(P))$, where $p(\cdot;\cdot)$ and $d(\cdot;\cdot)$ are defined as in \eqref{eq:primalObj} and \eqref{eq:dualObj}, and $J^*(P)$ is the optimal value as defined in \eqref{eq:MPC_primal}.
\end{proposition}
\begin{proof}
    By assumption, $\tilde U_{\theta_p^*}(P)$ and $\tilde\lambda_{\theta_d^*}(P)$ are primal and dual feasible, respectively. 
    Hence, $J^*(P) \leq p(P; \tilde U_{\theta_p^*}(P))$ and  $d(P;\tilde\lambda_{\theta_d^*}(P)) \leq D^*(P)$.
    By weak duality (Section~\ref{sec:dualitySection}), $D^*(P) \leq J^*(P)$, which concludes the proof.
\end{proof}

Proposition~\ref{prop:PDcheck} allows us to (conservatively) upper bound the suboptimality of the approximated primal policy $\tilde U_{\theta_p^*}(\cdot)$  using the approximated dual policy $\tilde\lambda_{\theta_d^*}(\cdot)$. Since computing $\tilde\lambda_{\theta_d^*}(P_t)$ merely evolves simple arithmetic operations,  Proposition~\ref{prop:PDcheck} provides a computationally efficient way to estimate the suboptimality of the approximated primal policy $\tilde U_{\theta_p^*}(P_t)$ for any parameter $P_t\in\mathcal P$, 
without solving the optimization problem \eqref{eq:MPCproblemOrig}.

\subsection{Primal-Dual Policy Learning Algorithm}
We are now ready to present our \emph{Primal-Dual Policy Learning} algorithm, which consist of two phases:
An offline \emph{learning phase}, where the primal and dual policies are trained; and an online \emph{control process} phase, where the performance (``optimality") of the approximated MPC policy $\tilde U_{\theta^*_p}(P_t)$ is verified in real-time  and for each parameter $P_t$ by performing the duality check \eqref{eq:PDcheck}.
During the online check, if the degradation in performance is less than a user-defined level $t_\textnormal{max}$, then the input $\tilde U_{\theta^*_p}(P_t)$ is applied; otherwise, a backup controller is invoked (details are below), and the procedure is repeated at the next time step.
Algorithm~\ref{alg:PDpolicyLearning} summarizes our scheme.

\begin{algorithm}[h!]
\caption{Primal-Dual Policy Learning}\label{alg:PDpolicyLearning}
\begin{algorithmic}[1]
\vspace{0.2em}

\Statex \hspace{-1.05em}
\textbf{Input:} max. backup-call frequency $\epsilon>0$, confidence level $0<\beta\ll 1$, suboptimality level $t_\textnormal{max}>0$.

\Statex \hspace{-1.05em}
\textbf{Select:} $\epsilon_p,\epsilon_d > 0$, $\epsilon_p + \epsilon_d = \epsilon$; $\beta_p, \beta_d > 0$, $\beta_p + \beta_d = \beta$.
\Statex{} 

\Statex\hspace{-1.4em} \textit{Learning Process (offline)}
\vspace{0.2em}
\Statex\hspace{-1em}\textbf{begin training}

\vspace{0.1em}
\hspace{-0.7cm}
\begin{minipage}[t]{\linewidth}
\begin{algorithmic}[1]
\State {Choose approximation \eqref{eq:primalPolicyParam_linear}-\eqref{eq:dualPolicyParam_linear} or \eqref{eq:primalPolicyParam_DNN}-\eqref{eq:dualPolicyParam_DNN}};
\State {Obtain $N_p$ and $N_d$ from Prop.~\ref{prop:scenarioOptSampleSize} or Prop.~\ref{prop:DNNSampleSize}};
\State Learn primal policy $\tilde U_{\theta_p^*}(\cdot) \approx U^*(\cdot)$ as in \eqref{eq:primalLearning};
\State Learn dual policy $\tilde{\lambda}_{\theta_d^*} \approx {\lambda}^*(\cdot)$ as in \eqref{eq:dualLearning};
\State Repeat until $t_d^* + t_p^* \leq t_\textnormal{max}$;
\end{algorithmic}
\end{minipage}
\vspace{0.01em}

\Statex \hspace{-1em}\textbf{end training} 
\Statex{} 
\Statex\hspace{-1.4em} \textit{Control Process (online)}
\vspace{0.2em}
\Statex\hspace{-1em}\textbf{begin control loop} (for $t=0,1,\ldots$)

\vspace{0.1em}
\hspace{-0.7cm}
\begin{minipage}[t]{\linewidth}
\begin{algorithmic}[1]
\State Obtain $P_t$ and evaluate $\tilde U_{\theta_p^*}(P_t)$ and $\tilde\lambda_{\theta_d^*}(P_t)$;
\State Verify $H(P_t)\tilde U_{\theta_p^*}(P_t)\leq h(P_t)$ and $\tilde\lambda_{\theta_d^*}(P_t)\geq0$; 
\State If feasible and \eqref{eq:PDcheck} $\leq t_\textnormal{max}$, then apply first element of $\tilde U_{\theta_p^*}(P_t)$;
\State Else, apply back-up controller.
\end{algorithmic}
\end{minipage}
\vspace{0.01em}
\Statex \hspace{-1em}\textbf{end control loop} {(until end of process)}
\end{algorithmic}
\end{algorithm}
{
Notice that feasibility, i.e. safety, is checked in Step~2 of the control process, while optimality, i.e., performance, is ensured in Step~3 of the control process.
}
In practice it is often desirable to bound the frequency of which the backup controller is invoked since, depending on the choice of the backup controller, calling it might result in reduced performance and/or increased computation time. The following theorem precisely links the samples sizes $N_p$ and $N_d$ in Algorithm~1 with the probability that the backup controller will be applied.

\begin{theorem}\label{thm:verification}
    Assume that $\tilde U_{\theta_p^*}(P_t)$ and $\tilde\lambda_{\theta_d^*}(P_t)$ take the form \eqref{eq:primalPolicyParam_linear}--\eqref{eq:dualPolicyParam_linear} or \eqref{eq:primalPolicyParam_DNN}--\eqref{eq:dualPolicyParam_DNN}. Let $\epsilon\in(0,1)$ be  
    the maximum admissible probability of the backup controller being invoked, $0 < \beta \ll 1$ a desired confidence level, and $\epsilon_p,\epsilon_d,\beta_p,\beta_d\in(0,1)$ be chosen as in Algorithm~\ref{alg:PDpolicyLearning}. If the primal and dual sample sizes $N_p$ and $N_d$ are chosen as in Propositions~\ref{prop:scenarioOptSampleSize} and \ref{prop:DNNSampleSize}, respectively, and if $t_d^*+t_p^* \leq t_\textnormal{max}$, then, with confidence at least $1-\beta$, it holds:
    \begin{align*}
         \mathbb{P}\big[ \hspace{.2cm} & H(P) \tilde U_{\theta_p^*}(P)\leq h(P), \tilde \lambda_{\theta_d^*}(P)\geq0, \\
           & p(P;\tilde U_{\theta_p^*}(P)) - d(P;\tilde\lambda_{\theta_d^*}(P))  \leq t_\textnormal{max} \hspace{.2cm} \big] \geq 1-\epsilon. 
    \end{align*}
\end{theorem}
\begin{proof}
   Application of the union bound \cite{bonferroni1936teoria} and Propositions~\ref{prop:scenarioOptSampleSize}, \ref{prop:DNNSampleSize} and \ref{prop:PDcheck} immediately lead to the desired result.
\end{proof}
Theorem~\ref{thm:verification} ensures that, if the primal and dual policies are trained with a ``sufficiently large" training data set, then the primal and dual policies will be often ``close" to each other, measured in terms of the objective function, and the backup controller will be invoked with probability less than $\epsilon$. Since \eqref{eq:PDcheck} constitutes an upper bound on the primal policy's suboptimality, it follows immediately that the primal policy $\tilde U_{\theta_p^*}$ is feasible and at most $t_\textnormal{max}$-suboptimal with probability, as formalized in the following corollary.
\begin{corollary}\label{cor:PDpolicyLearning}
    With confidence at least $1-\beta$, the control policy $\tilde U_{\theta_p^*}(\cdot)$ obtained from Algorithm~1 is feasible and at most $t_\textnormal{max}$-suboptimal with probability at least $1-\epsilon$. That is, with confidence at least $1-\beta$, we have $\mathbb{P}\big[ H(P) \tilde U_{\theta_p^*}(P)\leq h(P), p(P;\tilde U_{\theta_p^*}(P))  \leq J^*(P) + t_\textnormal{max} \big] \geq 1-\epsilon$.
\end{corollary}

In the following, we discuss aspects and extensions to our primal-dual policy learning framework (Algorithm~1).

\subsection{Discussion}\label{sec:discussion}
\subsubsection{Iterative Offline Learning Process}

The learning/training process proposed in Algorithm~1 is iterative in nature. This is because the suboptimality estimates $t^*_p$ and $t^*_d$ are only known after solving the primal/dual learning problems \eqref{eq:primalLearning}--\eqref{eq:dualLearning}. Indeed, if the approximated policies are not ``good enough", i.e., $t^*_p+t^*_d > t_\textnormal{max}$, where $t_\textnormal{max}$ is the user-defined maximal admissible suboptimality level, then the policies need to be retrained using a parametrization that captures a richer class of functions. In case a neural network parametrization of form \eqref{eq:primalPolicyParam_DNN}--\eqref{eq:dualPolicyParam_DNN} is used, one way to increase the richness (``expressiveness") of the neural network is to increase the depth and/or width of the network. Notice however, whenever the parametrization is altered, the sample sizes $N_p$ and $N_d$ need to be adapted accordingly.


We point out that, instead of using the proposed primal/dual learning problems \eqref{eq:primalLearning}--\eqref{eq:dualLearning}, alternative training and verification methods can be used for the offline phase. One such approach was presented in our previous work \cite{zhang2019safe}, where the primal-dual policies are not synthesized using predefined sample sizes $N_p$ and $N_d$, but obtained by solving simpler unconstrained supervised learning problems which minimize the mean squared error, e.g., $\min_{\theta_p}\sum_{i=1}^{N} \| \tilde{U}_{\theta_p}(P^{(i)}) - U^*(P^{(i)}) \|$ and $\min_{\theta_d} \sum_{i=1}^{N} \| \tilde{\lambda}_{\theta_d}(P^{(i)}) - \lambda^*(P^{(i)}) \|$. Performance (safety and optimality) of the trained policies are then ensured using a sampling-based verification scheme, see  \cite{zhang2019safe} for details.

\subsubsection{Computational Complexity}
Training the policies during the learning phase is typically computationally more demanding than evaluating the policies during the control process, since it requires solving the sampled optimization problems \eqref{eq:primalLearning}--\eqref{eq:dualLearning}. 
These optimization problems can consist of a large number of constraints and can even be non-convex, for example when parameterization \eqref{eq:primalPolicyParam_DNN}--\eqref{eq:dualPolicyParam_DNN} is used. Fortunately, the training can be carried out offline, and given today's availability in computation power, numerical optimization algorithms and toolboxes, it is generally possible to (at least approximately) solve problems \eqref{eq:primalLearning}--\eqref{eq:dualLearning}. 

Contrary to training, evaluating the approximated policies $\tilde{U}_{\theta_p^*}(\cdot)$ and $\tilde\lambda_{\theta_d^*}(\cdot)$ and computing the gap \eqref{eq:PDcheck} can be carried out very efficiently since they only consists of (straight-forward) functions evaluations. This is in contrast to solving the MPC problem \eqref{eq:MPCproblemOrig} online, which typically involves a (non-trivial) iterative algorithm that may require matrix inversions in each step.
In summary, our policy learning framework moves the ``burden" of solving the optimization problem from online to offline, i.e., from a time-critical part to a non-time-critical part.

\subsubsection{Choice of Policy Parametrization}
The exact choice of the basis functions in \eqref{eq:primalPolicyParam_linear}--\eqref{eq:dualPolicyParam_linear}, as well as the size, depth and type of activation function in the neural network approximation   \eqref{eq:primalPolicyParam_DNN}--\eqref{eq:dualPolicyParam_DNN} is highly problem-dependent, and should be experimented with in practice. Furthermore, our primal-dual policy learning scheme is not restricted to policies of the forms  and \eqref{eq:primalPolicyParam_DNN}--\eqref{eq:dualPolicyParam_DNN}, but can also accommodate for other parametrizations. This is because, regardless of their representations and assuming feasibility of $\tilde{U}_{\theta_p^*}(P)$ and $\tilde{\lambda}_{\theta_d^*}(P)$ , the dual cost $d(P;\tilde{\lambda}_{\theta_d^*}(P))$ is always lower bounds the primal cost $p(P;\tilde{U}_{\theta_p^*}(P))$, for all $P\in\mathcal{P}$. Indeed, any function parametrization, as long as  upper bounds on the sample sizes $N_p$ and $N_d$ can be established, can be used, given freedom to the control engineer to tailor the policy representation to the problem at hand.

\subsubsection{Backup Controller}
Similar to the choice of the policy parameterization, the choice of the backup controller, which is used in Algorithm~1 as a fall-back strategy whenever the primal policy cannot to be certified to be feasible and less than $t_\textnormal{max}$-suboptimal, is highly problem-dependent. For example, a conservatively tuned PID or LQR controller can be used which sacrifices performance for constraint satisfaction. Alternatively, one may chose to solve the MPC problem each time the check fails and use the solutions from the primal and dual policies to warm-start the numerical solver. 
{
While further investigation is needed, preliminary analysis  shows that such warm-starting numerical solvers can reduce the computation time by up to 50\%.
}

\section{Case Study: Integrated Chassis Control}\label{sec:simResults}
In this section we present a numerical case study in the field of integrated chassis control (ICC) for vehicles. Roughly speaking, the goal in chassis control is to improve a vehicle's dynamics and user comfort by actively controlling a vehicle's motions. Traditionally, chassis control is carried out by many independent subsystems. More recently however, there have been attempts to consider the coupling between the individual subsystems, giving rise to so-called integrated chassis control (ICC) \cite{ChandrasekanICC2011}. To coordinate the whole system, a multi-layer structure has been proposed in the literature, which first computes the desired forces and moments at the vehicle's center of gravity in the upper-layer, and then distributes them to each subsystem in the lower-layer. 
Due to this top-down strategy, the upper-layer controller needs to be developed more comprehensively. This is typically achieved by incorporating a vehicle model that considers both state coupling as well as state and input constraints.
Hence, MPC is a natural control strategy, since it allows the incorporation of input and state constraints in a disciplined manner, and is able to handle the multivariate nature of the task.

In the following we study the problem of calculating the yaw moment, the roll moment, and the lateral force in the upper level 
of an ICC structure. Specifically, given the driver's current inputs and a reference trajectory, the objective is to find the optimal torques and forces that is generated by the MPC controller. Unfortunately, it turns out that incorporating MPC into ICC on a real vehicle is a challenging task, because the computation power is very limited and yet the systems should be operated at least at 100~Hz.
In the following, we demonstrate how the proposed primal-dual policy learning scheme can be used to help enable MPC on such resource-constrained fast systems.

\subsection{Problem Formulation}
We consider a linear parameter-varying system of the form
\begin{align}\label{eq:icc_model_simple}
 x_{t+1} = A(v_t)x_t + B(v_t)u_t + E(v_t) \delta_t,\quad y_t = C x_t,  
\end{align}
where $x := [\beta, r, \phi, \dot{\phi}] \in \R^4$ is the state, $u := [\Delta M^z, \Delta M^x, \Delta F^y]\in\R^3$ is the input, $y_t \in \mathbb{R}^3$ is the output, $v_t$ is the vehicle's longitudinal velocity, and $\delta_t$ is the driver's steering input, which we assume can be predicted, see Table~\ref{tab:stateInputVariables}. 
The system matrices  are given in the Appendix.

\begin{table}[h!]
\caption{Variables used in \eqref{eq:icc_model_simple}.}
\label{tab:stateInputVariables}
\centering 
\ra{1.3}
\begin{tabular}{@{}l l | l l @{}}\toprule
variable & description & variable & description\\
 \midrule
$\beta$ & side-slip angle   & $\Delta M^z$  & yaw moment            \\
$r$     & yaw rate          & $\Delta M^x$  & roll moment          \\
$\phi$  & roll angle        & $\Delta F^y$  & lateral force          \\
$\dot\phi$  & roll rate     & $v$           & long.\ velocity          \\
$\delta$    & driver steering angle & $T=3$       & pred.\ horizon               \\
\bottomrule
\end{tabular}
\end{table} 

The control objective is to minimize the output tracking error while satisfying input constraints and input rate constraints. Hence, the MPC problem is given by
\begin{equation}\label{eq:ICC_example}
    \begin{array}{llll}
        \displaystyle \min_{X_t,U_t} & \displaystyle \sum_{k=0}^{T-1} (y_{k|t}-y_{k|t}^\text{ref})^\top Q_s (y_{k|t}-y_{k|t}^\text{ref}) + u_{k|t}^\top R_s u_{k|t}  \\
        \ \ \text{s.t.}   & x_{k+1|t} = A(v_{t}) x_{k|t} + B(v_{t}) u_{k|t} + E(v_{t})\delta_{k|t}, \\
         & y_{k|t} = C x_{k|t}, \ |u_{k|t}|\leq \bar{u},\ |u_{k|t} - u_{k-1|t}| \leq \xoverline{\Delta u}, \\
         & x_{0|t} = x_t,\ u_{-1|t} = u_{t-1},\ k=0,\ldots,T-1,
    \end{array}
\end{equation}
where $Q_s,R_s \in \mathbb{S}^{3 \times 3}_+$ are positive definite matrices, $y_{k|t}^\text{ref}\in\R^3$ are given reference signals, $\bar{u}$ defines the input constraints, $\xoverline{\Delta u}$ defines the rate constraints, and $u_{t-1}$ is the previous input. The parameters in \eqref{eq:ICC_example} are $P_t := (x_t, v_t, \{y_{k|t}^\text{ref},\delta_{k|t}\}_{k}, u_{-1|t})\in \mathcal{P} \subset\R^{19}$, where $\mathcal{P}$ is chosen appropriately.
Since the velocity $v_t$ enters the dynamics in a nonlinear fashion (see Appendix), the optimal control law $U^*(P_t)$ cannot be derived using standard methods from explicit MPC. In the following, we approximate $U^*(\cdot)$ using the approach described in Section~\ref{sec:primalDualPolicyLearning}. Throughout, we consider a horizon of $T=3$ such that $U^*(\cdot): \R^{19} \to \R^9$ and $\lambda^*(\cdot): \R^{19} \to \R^{36}$. 

\subsection{Training and Verification} \label{sec:trainnets}
In this section, we illustrate our proposed primal-dual policy learning method on \eqref{eq:ICC_example} for the case when $(i)$ the policies are approximated using a weighted sum of radial basis functions (``radial basis network, RBN"), and $(ii)$ when a neural network is used.
For both cases, the following parameters were chosen for  Algorithm~1: $\epsilon = 0.1$, $\beta=2\cdot10^{-7}$, $t_\textnormal{max} = 4$, and $\epsilon_p = \epsilon_d = \epsilon/2$, $\beta_p = \beta_d = \beta/2$. The training data $P^{(i)}$ are sampled uniformly from $\mathcal{P}$. We next briefly describe our training procedure for both parametrizations

\subsubsection{Weighted Sum of Basis Functions} 
Inspired by \cite{DomahidiLearning2011}, we consider the following primal and dual radial basis functions 
\begin{align*}
& \kappa_{p, k} (P) = ( \kappa_p(P)^\top \otimes \mathbb{I}_9) [\cdot,k],~k = 1,2,\dots, L_p, \\
& \kappa_{d, k} (P) = ( \kappa_d(P)^\top \otimes \mathbb{I}_{36}) [\cdot,k],~k = 1,2,\dots, L_d, \\
\end{align*}
where $M[\cdot,k]$ denotes the $k$-th column of a matrix $M$, $\kappa_p(P) := [ \hat\kappa_{p,1}(P) , \ldots, \hat\kappa_{p,n_p^\textnormal{rb}}(P) ]$, $\kappa_d(P) := [ \hat\kappa_{d,1}(P) , \ldots, \hat\kappa_{d,n_d^\textnormal{rb}}(P) ]$, and
\begin{align*}
& \hat\kappa_{p,j}(P) = \left ( 1- \frac{\Vert W_s(P - P_c^{(j)}) \Vert_2}{2} \right ) ^3, ~j=1,\ldots,n^{\textnormal{rb}}_{p}, \\
& \hat\kappa_{d,j}(P) = \left ( 1 - \frac{\Vert W_s(P - P_c^{(j)}) \Vert_2}{2} \right )^3, ~j=1,\ldots,n^{\textnormal{rb}}_{d},
 \end{align*}
where $W_s$ is a problem-dependent (diagonal) scaling matrix.
Above, $P_c^{(j)}$ denotes the center of the $j$-th radial basis function, and $L_p = 9n^\textnormal{rb}_{p}$ and $ L_d = 36n^\textnormal{rb}_{d}$ are the number of radial basis functions for primal and the dual RBNs respectively. 
The centers of basis functions $P_c^{(j)}$ for $j = 1,2,\dots,n^\textnormal{rb}_{p/d}$ were chosen at random. 
We obtain the sample sizes for solving \eqref{eq:primalLearning}--\eqref{eq:dualLearning} using Proposition~1. 
We start from $n^\textnormal{rb}_p = n^\textnormal{rb}_{d} = 10$ and increment each gradually by $10$ after each failed training attempt. The obtained final values are: $n^\textnormal{rb}_p = 130, n^\textnormal{rb}_{d} = 100, 
t_p^* = 2.703$ and $t_d^* = 0.7256$, satisfying $t_p^* + t_d^* \leq t_\textnormal{max}$.

\subsubsection{Deep Neural Network}
Motivated by \cite{ChenMorariACC2018}, we use a deep neural network with rectified linear unit  activation functions (DNN-ReLU), with training sample sizes from Proposition~\ref{prop:DNNSampleSize}. During the training phase, we start with a neural network of width 5 and depth 3, and increment its width gradually by 5 neurons after each failed training attempt, until $t_p^*+t_d^* \leq t_\textnormal{max}$. The final network size we obtained was of width 15 for the primal network $\tilde U_{\theta_p^*}(\cdot)$, and of width 5 for the dual network $\tilde{\lambda}_{\theta_d^*}(\cdot)$,
such that $t_p^* = 0.3764$ and $t_d^* = 1.1347$.
Hence, by virtue of \eqref{eq:PDcheck}, the approximated MPC policy $\tilde U_{\theta_p^*}(\cdot)$ will be at most 1.5111 suboptimal, with probability at least $1-\epsilon$.

\subsubsection{Monte Carlo Simulations}
In this section, we are interested in evaluating the performance and conservatism of the approximated policies. To this end, we evaluate the empirical suboptimality levels $\hat t_p(P) := p(P; \tilde U_{\theta_p^*}(P)) - J^*(P)$, $\hat t_d(P) := J^*(P) - d(P;\tilde \lambda_{\theta_d^*}(P))$ and $\hat t(P) := p(P;\tilde U_{\theta_p^*}(P)) - d(P;\tilde\lambda_{\theta_d^*}(P))$ for 1'000'000 random parameters $P^{(i)}$, and also determine the empirical violation probabilities $\hat \epsilon_p$, $\hat \epsilon_d$ and $\hat\epsilon$. These values are reported in Table~\ref{tab:subOpt_Convex} for RBN, and in Table~\ref{tab:subOpt} for DNN-ReLU.


\begin{table}[h!]
\caption{Empirical suboptimality levels and violation probabilities for RBN. }
\label{tab:subOpt_Convex}
\centering 
\ra{1.3}
\begin{tabular}{@{}l | l l l || c c c @{}}\toprule
 & $\hat t_p$ & $\hat t_d$ & $\hat t$ & $\hat\epsilon_p$ & $\hat\epsilon_d$ & $\hat\epsilon$ \\
 \midrule
 mean       & 1.768     & 0.5957       & 2.3644        &   &    &  \\ 
 median     & 1.907  & 0.6667  & 2.5451        & 0.412\%   & 2.315\%   & 1.452\% \\ 
 maximum    &  5.044    & 2.7727       & 6.1412        &           &       &  \\
\bottomrule
\end{tabular}
\end{table} 

\begin{table}[h!]
\caption{Empirical suboptimality levels and violation probabilities for DNN-ReLU. }
\label{tab:subOpt}
\centering 
\ra{1.3}
\begin{tabular}{@{}l | l l l || c c c @{}}\toprule
 & $\hat t_p$ & $\hat t_d$ & $\hat t$ & $\hat\epsilon_p$ & $\hat\epsilon_d$ & $\hat\epsilon$ \\
 \midrule
 mean       & 0.0034     & 0.0495       & 0.053        &   &    &  \\ 
 median     & 8.06 $\cdot 10^{-5}$    & 0.0191       & 0.020        & 0.012\%   & 0.044\%   & 0.001\% \\ 
 maximum    &  2.0842    & 4.1210       & 4.216        &           &       &  \\
\bottomrule
\end{tabular}
\end{table} 

We see from the tables that the approximated primal and dual policies result in control laws that are generally close to optimal, with a median duality gap of $\hat t = 2.5451$ (RBN) and $\hat t = 0.020$ (DNN-ReLU). Furthermore, we see that while the worst-case duality gap  $\hat t$ of both RBN and DNN-ReLU exceed $t_\textnormal{max}=4$, this occurs with probability 1.45\% (RBN) and 0.001\% (DNN-ReLU), which is significantly lower than the targeted value of $\epsilon=0.1$. 
Therefore, statistically speaking, the backup controller in Algorithm~1 will only be called $0.001\%$-fraction of all times when the DNN-ReLU controller is deployed. The results from Table~\ref{tab:subOpt_Convex} and Table~\ref{tab:subOpt} imply that, for this numerical example, the sample sizes provided by Proposition~\ref{prop:scenarioOptSampleSize} and Proposition~\ref{prop:DNNSampleSize} are conservative; in practice, the trained policies perform better than expected, both in terms of violation probability and suboptimality level, especially when a neural network is used. 


\subsection{Comparison with Online MPC and Explicit MPC}
In this section, we compare the approximated MPC controllers with online MPC and a suboptimal explicit MPC, both in terms of computational complexity and optimality of the generated inputs. 
To solve online MPC, we use Mosek \cite{mosek} and Gurobi \cite{optimization2012gurobi}, two state-of-the-art numerical solvers.
In terms of explicit MPC, recall that standard explicit MPC method cannot be applied to \eqref{eq:ICC_example} because the longitudinal velocity $v$ enters the system matrices $A$ and $B$ parameterically. To circumvent this issue, we discretize this velocity parameter $v$, and compute an explicit MPC law $U_{v}^*(\cdot)$ for each value of $v$ using the MPT toolbox \cite{MPT3}\footnote{We notice that with $v$ fixed, the resulting controller $U_{v}^*(\hat P_t)$ is polyhedral piecewise affine in the remaining parameters $\hat P_t = (x_t,  \{y_{k|t}^\text{ref},\delta_{k|t}\}_{k}, u_{-1|t})\in\R^{18}$.}. Due to this approximation, we call it \emph{suboptimal explicit MPC}. The control law can now be looked up as follows: For every parameter $P\in\mathcal{P}$, we first extract the associated velocity parameter, and evaluate the explicit MPC whose velocity is closest to the above.


\subsubsection{Computation Time}

Table~\ref{tab:compTime} reports the computation time of our primal-dual policy approximation scheme using RBN and DNN-ReLU, the online MPC, as well as the suboptimal explicit MPC, evaluated on $10'000$ randomly extracted test samples. The timings are taken on an early 2016 MacBook, that runs on a 1.3 GHz Intel Core m7 processor and is equipped with 8 GB RAM and 512 GB SSD. We use MATLAB generated MEX files to determine the run-time of RBN, DNN-ReLU and suboptimal explicit MPC. The optimization problems are formulated in YALMIP \cite{lofberg:05}, and the timings are those reported by the solvers themselves.

\begin{table}[h!]
\caption{Computation Times, rounded to two significant digits.}
\label{tab:compTime}
\centering 
\ra{1.3}
\begin{tabular}{@{}l | l  l |  l  l | l @{}}\toprule
time [ms] & RBN & DNN-ReLU  & Gurobi & Mosek  & subopt. EMPC  \\
 \midrule
 min.   &  0.172  &  0.024  & 1.4  & 1.5   & 0.014 \\ 
 max.   & 0.386   &  0.056  & 2.5  & 3.0   & 0.053 \\
 mean   &  0.206  &  0.029  & 1.6  & 1.8   & 0.018 \\
 std.   & 0.034   &  0.006  & 0.16 & 0.25  & 0.0039 \\
\bottomrule
\end{tabular}
\end{table} 

From Table~\ref{tab:compTime}, we see that RBN and DNN-ReLU are significantly faster than online MPC. On average, the RBN is $\sim \!\! 7.7\times$ faster than Gurobi and $\sim \!\! 8.7\times$ faster than Mosek, while DNN-ReLU is $\sim \!\! 55\times$ faster than Gurobi, and $\sim \!\! 62\times$ faster than Mosek.
This is because evaluating an RBN and DNN-ReLU just involves simple function evaluations and  matrix-vector multiplications, but no (iterative) matrix inversions, as opposed to the case of numerical optimization. 
Moreover, we observe that the evaluation times obtained using DNN-ReLU are more consistent than those of Gurobi and Mosek, with a standard deviation of~10\%. 
Thus by replacing an online solver with DNN, we can achieve significant speedup, although "training of DNN" can be very expensive. 
On the other hand, notice that the suboptimal explicit MPC controller is significantly faster than both policy approximation schemes and online MPC, with an average execution time of 0.018~ms. This is because each suboptimal explicit MPC $U_v^*(\cdot)$ only has 9 polyhedral regions, allowing rapid evaluation of the piecewise-affine control laws.


\subsubsection{Suboptimality}\label{ssec:memory_sub}
In this section, we compare the suboptimality level of our approximated MPC controllers (RBN and DNN-ReLU) with that of the suboptimal explicit MPC (EMPC) \footnote{Online MPC is not considered as it always solves for the optimal control input.}. As a metric, we consider the \emph{relative} suboptimality
\begin{align*}
    & \hat{t}_p^\text{rel}(P;\tilde{U}) :=  \frac{p(P;\tilde{U}(P)) - J^*(P)}{J^*(P)}. 
\end{align*}
where  $J^*(P)$ is the optimal value obtained by solving \eqref{eq:ICC_example} with parameter $P$, and $\tilde U(P)$ is either given by the RBN/DNN-ReLU MPC or the suboptimal explicit MPC. The relative primal suboptimality levels are listed in Table~\ref{tab:subOpt_expl}, with two different velocity discretization schemes for suboptimal explicit MPC, $\Delta v = 1$ m/s and $\Delta v = 0.25$ m/s. 


\begin{table}[h!]
\caption{Relative primal sub-optimality of RBN \& DNN-ReLU, and suboptimal explicit MPC approximations, evaluated on 10'000 samples.}
\label{tab:subOpt_expl}
\centering 
\ra{1.3}
\begin{tabular}{@{}l | l  | l | l l @{}} \toprule
$\hat t_p^\text{rel}$ & RBN & DNN-ReLU  & subopt. EMPC & subopt. EMPC\\
 &   &  & ($\Delta v = 1$ m/s)  & $(\Delta v = 0.25$ m/s)\\
 \midrule
 min.    & $3.15\cdot10^{-7}$ & $5.76\cdot10^{-10}$    & $5.96\cdot 10^{-7}$ & $3.10\cdot 10^{-6}$      \\ 
 max.    & $0.0626$ & $3.01\cdot10^{-3}$    &  $0.917$ &  $0.9033$     \\
 mean   &  $0.0049$  & $3.395\cdot10^{-5}$     & 0.128    &  $0.120$    \\
 std.  &  $0.0039$ & $1.03\cdot10^{-4}$  &  0.201 &  $0.2020$
 \\
\bottomrule
\end{tabular}
\end{table} 

We see from Table~\ref{tab:subOpt_expl} that the RBN and DNN-ReLU approximations yield significantly better performance than suboptimal explicit MPC, with a maximum suboptimality level of 6\% (RBN) and 0.3\% (DNN-ReLU), compared to 91.7\% when the suboptimal EMPC is deployed. Furthermore, the relative suboptimality of RBN and DNN-ReLU controllers are, on average, $\sim \!\! 26\times$ and $\sim \!\! 3800\times$ lower than that of suboptimal EMPC. Indeed, with an average suboptimality level of $3.4\cdot 10^{-5}$, the DNN-ReLU controller is essentially optimal for all practical purposes. It is also interesting to see that, for this example, gridding the velocity parameter $v$ finer in suboptimal EMPC only yielded insignificant performance improvements, with the average suboptimality level dropping from 12.8\% to 12.0\%. 

We close this comparison by pointing out that storing the suboptimal EMPC controller also requires significant memory. Indeed, for our problem, each $U_{ v}^*(\cdot)$ is a piecewise affine controller defined on nine polyhedral regions, and requires approximately 47.5~KB memory to be stored. If the velocity is gridded with  $\Delta v =1$ m/s from $v=2,\ldots,25$ m/s, then 1.14~MB memory is needed just to store the regions and corresponding control laws for suboptimal EMPC. If gridded with $\Delta v = 0.25$ m/s, then the space required to store the suboptimal explicit MPC increases to 4.56~MB, which is more than {four times of} what is available on our targeted automative-grade electronic control unit (ECU), see below.

\subsection{Experimental Results on  Automotive ECU}\label{ssec:ecu_impl}
In this section we present experimental results of our proposed controllers. Specifically, we implement Algorithm~\ref{alg:PDpolicyLearning} with the above discussed DNN-ReLU on an automotive-grade Infineon AURIX family TC27x series ECU.
We chose DNN-ReLU over RBN due to DNN-ReLU's superior performance both in terms of optimality (Table~\ref{tab:subOpt_expl}) and computation time (Table~\ref{tab:compTime}) during our numerical simulations.  
The considered electronic control unit (ECU) has a clock speed of 200 MHz and a total of 4~MB ROM, of which approximately 1~MB is available for the controller. The requirement is to run the MPC controller at least at 50~Hz, ideally at 100~Hz. 
Table~\ref{tab:ECU_data} compares our DNN-ReLU implementation with online MPC, where we use the state-of-the-art high-speed solver CVXGEN \cite{MattingleyBoydCVXGen2012} to solve the optimization problem \eqref{eq:ICC_example} online. To satisfy the 50~Hz requirement, the number of iterations CVXGEN performs is limited to 23.
Suboptimal explicit MPC was not implemented due to its higher suboptimality levels (Table~\ref{tab:subOpt}) and, most importantly, large memory requirement (Section~\ref{ssec:memory_sub}). 
\begin{table}[h!]
\caption{Average computation time and storage requirements on ECU.}
\label{tab:ECU_data}
\centering 
\ra{1.3}
\begin{tabular}{@{}l | l  l l @{}}\toprule
Method & time [ms]  & memory [KB] \\
 \midrule
 DNN-ReLU   & 1.8     & 114 \\   
 CVXGEN     &  18.7     & 322 \\ 
\bottomrule
\end{tabular}
\end{table} 



We see from Table~\ref{tab:ECU_data} that, with an average computation time\footnote{Due to technical limitations, we were only able to determine the average computation time.} of 1.8~ms, DNN-ReLU is about $\sim \!\! 10\times$ faster than CVXGEN, which requires on average 18.7~ms to solve the MPC problem \eqref{eq:ICC_example}. In fact, since the requirement is to run the controller at least at 100~Hz, online MPC with CVXGEN is not real-time feasible without changing the ECU. 

To experimentally evaluate the performance of the controller, we deployed the DNN-ReLU controller onto the Berkeley Autonomous Car, a retrofitted 2013 Hyundai Genesis G380. The experiment, which consists of a single lane change over a period of approximately 12~seconds, was carried out at the Hyundai-Kia's California Proving Ground. A robot performs the lane change, during which the vehicle is accelerated from 3m/s to 21.5m/s. Our analysis revealed that, during our experiments, the backup controller  never had to be activated. Moreover, the average relative primal suboptimality level  for our DNN-ReLU controller was $1.1\%$, rendering it near-optimal.

We conclude this section by pointing out that the proposed primal-dual policy learning scheme is a promising method to generate fast and near-optimal MPC policies.

\section{Conclusion}\label{sec:conclusion}
In this paper, we proposed a new method for approximating the explicit MPC control policy for linear parameter varying systems using supervised learning techniques. The feasibility and optimality of the approximated control policy is ensured with high probability using the theory of randomized optimization.
Furthermore, we introduced a novel \emph{dual policy}-based verification scheme that certifies the optimality of the approximated control policy during run-time. Since this proposed verification scheme only requires the evaluation of trained policies, our algorithm is computationally efficient, and can be implemented even on resource-constrained systems. Our numerical case study has revealed that then the proposed primal-dual policy learning framework allows the integrated chassis control problem to be solved 10x faster on an production-grade electronic control unit compared to the state-of-the-art numerical solver CVXGEN, while sacrificing minimal amount of performance, and maintaining safety of the vehicle.



\section*{ACKNOWLEDGMENT}
This research was partially funded by the Hyundai Center of Excellence at UC Berkeley. This work  was  also  sponsored  by  the  Office  of  Naval  Research (ONR-N00014-18-1-2833), and by Ford Motor Company. The authors thank Dr.~Yi-Wen Liao and Dr.~Jongsang Suh for helpful discussions on the ICC problem.

\begin{appendix}
\subsection*{System Modelling}\label{app:systemModel}
The system matrices are given as follows \cite{TAKANO2003149}: 
\begin{align*}
    A^c(v_t) & = \begin{bmatrix} -\frac{\bar{C}I_x}{fv_t} + \mathrm{dt} & -1+ \frac{I_x\bar{C}_1}{fv_t^2} & \frac{m_sh_s\bar{M}}{fv_t} & \frac{m_s h_s C_\phi}{fv_t} \\ \frac{\bar{C}_1}{I_z} & -\frac{\bar{C}_2}{I_zv_t}+\mathrm{dt} & 0 & 0\\ 0 & 0 & \mathrm{dt} & 1\\ -\frac{m_sh_s\bar{C}}{f} & -\frac{m_sh_s\bar{C}_1}{fv_t} & \frac{m\bar{M}}{f} & -\frac{mC_\phi}{f} + \mathrm{dt}  \end{bmatrix},\\
    B(v_t) & = \begin{bmatrix} 0 & \frac{m_sh_s}{fv_t} & \frac{I_x}{fv_t} \\ \frac{1}{I_z}& 0 & 0\\ 0 & 0& 0\\ 0 & \frac{m}{f} & -\frac{m_sh_s}{f}\end{bmatrix} \mathrm{dt},~
    E(v_t) = \begin{bmatrix} \frac{I_xC_\textnormal{front}}{fv_t}\\\frac{C_\textnormal{front}l_\textnormal{front}}{I_z}\\0\\-\frac{m_sh_sC_\textnormal{front}}{f} \end{bmatrix} \mathrm{dt},
\end{align*}
where $v_t$ is the longitudinal velocity, $\bar{C} = C_\textnormal{front} + C_\textnormal{rear}$, $\bar{C}_1 = l_\textnormal{rear} C_\textnormal{rear} - l_\textnormal{front} C_\textnormal{front}$, $\bar{C}_2 = C_\textnormal{front} l_\textnormal{front}^2 + C_\textnormal{rear} l_\textnormal{rear}^2$, $\bar{M} = m_s g h_s -k_\phi$ and $f = mI_x-m_s^2h_s^2$. In the above $m$ is the total mass, $h_s$ is the length of roll moment arm, $m_s$ is the sprung mass, $I_x$ is the roll moment of inertia, $I_z$ is the yaw moment of inertia, $C_\phi$ is the roll damping coefficient, $k_\phi$ is the roll stiffness, $g$ is the acceleration due to gravity, $l_\textnormal{front}$ and $l_\textnormal{rear}$ are distances from center of gravity to front and rear axles respectively, and $C_\textnormal{front}$ and $C_\textnormal{rear}$ are the cornering stiffnesses of the front and rear tires respectively, of the vehicle.

\end{appendix}

\bibliographystyle{IEEEtran}   
\bibliography{library_GXZ}

\end{document}